\documentclass[conference]{IEEEtran}
\IEEEoverridecommandlockouts
\usepackage{cite}
\usepackage{amsmath,amssymb,amsfonts,bbm,bm}
\usepackage{algorithm}
\usepackage{algpseudocode}
\usepackage{graphicx}
\usepackage{float}
\usepackage{subfigure}
\usepackage{textcomp}
\usepackage{xcolor}
\usepackage{float}
\usepackage{paralist}
\usepackage{multirow}
\usepackage{booktabs}
\usepackage[font=small]{caption}
\usepackage{subcaption}
\usepackage{graphicx}
\usepackage{tikz,pgfplots}
\usetikzlibrary{shapes.geometric, arrows, positioning}
\usepackage[nolist]{acronym}

\DeclareMathOperator*{\argmax}{argmax}

\newtheorem{theorem}{Theorem}

\newtheorem{definition}{Definition}

\newcommand{\mB}{\mathbf{B}}
\newcommand{\mGamma}{\mathbf{\Gamma}}

\newcommand{\bv}{\mathbf{v}}
\newcommand{\bu}{\mathbf{u}}

\newcommand{\bh}{\mathbf{h}}
\newcommand{\bs}{\mathbf{s}}

\newcommand{\be}{\mathbf{e}}
\newcommand{\bl}{\mathbf{l}}

\newcommand{\bomega}{\boldsymbol{\omega}}

\newcommand{\setC}{\mathbb{C}}

\newcommand{\setS}{\mathbb{S}}

\newcommand{\btph}{\tilde{\boldsymbol{\varphi}}}

\newcommand{\trp}{\mathsf{T}}
\newcommand{\her}{\mathsf{H}}
\newcommand{\set}[1]{\left\lbrace #1 \right\rbrace}

\newcommand{\brc}[1]{\left( #1 \right)}
\newcommand{\dbc}[1]{\left[ #1 \right]}
\newcommand{\norm}[1]{\left\Vert #1 \right\Vert}
\newcommand{\abs}[1]{\left\vert #1 \right\vert}
\newcommand{\diag}[1]{\mathrm{diag} \left\lbrace #1 \right\rbrace}
\newcommand{\Ex}[1]{\mathbb{E} \left\lbrace #1 \right\rbrace}

\def\BibTeX{{\rm B\kern-.05em{\sc i\kern-.025em b}\kern-.08em
    T\kern-.1667em\lower.7ex\hbox{E}\kern-.125emX}}

\begin{document}
\begin{acronym}
\acro{ap}[AP]{access point}
\acro{pass}[PASS]{\underline{p}inching-\underline{a}ntenna \underline{s}y\underline{s}tem}
\acro{mimo}[MIMO]{multiple-input multiple-output}
\acro{ota-fl}[OTA-FL]{over-the-air federated learning}
\acro{los}[LoS]{line-of-sight}
\acro{ps}[PS]{parameter server}
\acro{noma}[NOMA]{non-orthogonal multiple access}
\acro{irs}[IRS]{reconfigurable intelligent surface}
\acro{snr}[SNR]{signal-to-noise ratio}
\acro{sinr}[SINR]{signal-to-interference and noise ratio}
\acro{bcd}[BCD]{block coordinate descent}
\acro{rzf}[RZF]{regularized zero-forcing}
\acro{pas}[PAS]{pinching antenna system }
\acro{fp}[FP]{fractional programming}
\acro{mrt}[MRT]{maximal-ratio transmission}
\acro{zf}[ZF]{zero-forcing}
\acro{mse}[MSE]{mean squared error}
\acro{fedavg}[FedAvg]{federated averaging}
\acro{fl}[FL]{federated learning}
\acro{awgn}[AWGN]{additive white Gaussian noise}
\end{acronym}

\title{Energy-Efficient Over-the-Air Federated Learning via Pinching Antenna Systems}
\author{
\IEEEauthorblockN{Saba Asaad}
\IEEEauthorblockA{University of Toronto\\
saba.asaad@utoronto.ca}
\and
\IEEEauthorblockN{Ali Bereyhi}
\IEEEauthorblockA{University of Toronto\\
ali.bereyhi@utoronto.ca}
\thanks{This work was supported by the German Research Foundation (DFG).}
}

\maketitle

\begin{abstract}
Pinching antennas systems (PASSs) have recently been proposed as a novel flexible-antenna technology. These systems are implemented by attaching low-cost pinching elements to dielectric waveguides. As the direct link is bypassed through waveguides, PASSs can effectively compensate large-scale effects of the wireless channel. This work explores the potential gains of employing PASSs for over-the-air federated learning (OTA-FL). For a PASS-assisted server, we develop a low-complexity algorithmic approach, which jointly tunes the PASS parameters and schedules the mobile devices for minimal energy consumption in OTA-FL. We study the efficiency of the proposed design and compare it against the conventional OTA-FL setting with MIMO server. Numerical experiments demonstrate that using a single-waveguide PASS at the server within a moderately sized area, the required energy for model aggregation is drastically reduced as compared to the case with fully-digital MIMO server. This introduces PASS as a potential technology for energy-efficient distributed learning in next generations of wireless systems.
\end{abstract}


\section{Introduction}
\label{sec:intro} 

\Ac{fl} is a distributed learning paradigm, where resource-constrained edge devices collaboratively train a model with the help of a server \cite{pmlr-v54-mcmahan17a}. Although promising, implementing \ac{fl} in realistic networks deals with high communication overhead. 
Wireless systems leverage \emph{over-the-air computation} (AirComp) to address this challenge \cite{AirComp1}. AirComp exploits the superposition property of multiple-access channel to compute a desired function in the analog domain.
Building on this concept, \ac{ota-fl} has been proposed, where 
the server uses AirComp to aggregate directly over the air \cite{yang2020federated,bereyhi2023device}. 
Despite its communication efficiency, \ac{ota-fl} remains vulnerable to channel effects and synchronization errors, which impair its convergence \cite{AirCompError}. To overcome these issues, 
recent studies have explored various \textit{channel-enhancement technologies}. Techniques such as massive \ac{mimo}, \ac{irs}, and reconfigurable antennas use their physical degrees of freedom to control signal propagation and improve AirComp-assisted model aggregation \cite{mMIMO-FL,IRS-FL1,IRS-FL2,zhao2025movable}. The study in \cite{mMIMO-FL} invokes favorable propagation in massive \ac{mimo} systems to improve the aggregation quality. In \cite{IRS-FL2} and \cite{ni22FL}, passive beamforming via \acp{irs} is exploited to suppress model aggregation error. 
Movable-antenna technology \cite{zhu2023movable} is used in \cite{zhao2025movable} to tune the communication channel toward the server for successive superposition of parameters. While these techniques are able to utilize their degrees of freedom to suppress the destructive superposition in uplink channel, they suffer from the inherent drawback of \textit{high path-loss}. For instance, \acp{irs} suffer from double-attenuation due to passive reflection, and massive \ac{mimo} and movable-antennas can only combat small-scale fading.  

Recently, \textit{\acp{pass}} have been proposed as a novel technology for wireless transmission \cite{PASS1,PASS2}. This proposal has been motivated by the prototype demonstrated at DOCOMO in 2022 \cite{DOCOMO-PASS}. A \ac{pass} employs a dielectric waveguide as its backbone medium and radiates the signal towards mobile users through pinching elements that are coupled with the waveguide and can be dynamically activated or deactivated at any position \cite{ouyang2025array}. Unlike conventional flexible-antenna systems, a \ac{pass} employs an arbitrarily long waveguide that allows antennas to be positioned close to users, and thereby it establishes strong \ac{los} links \cite{MIMO-PASS,PASS-DL}. This enables a \ac{pass}-assisted transmitter to significantly reduce the impact of path-loss in the system. From a hardware viewpoint, \ac{pass} is also inexpensive and easy to implement, as it only requires a single radio-frequency chain per waveguide and is realized by coupling dielectric elements to the waveguide \cite{PASS2}. It can be regarded as a specific realization of reconfigurable antenna concept, offering greater flexibility and scalability than existing architectures. 
\par Motivated by promising gains of \ac{pass}, this work develops a novel \ac{ota-fl} design, which deploys a \ac{pass} at the server to enhance over-the-air model aggregation. Noting that \ac{pass} primarily combats with large-scale fading impacts, our design focuses on energy efficiency of \ac{ota-fl}. 
Our main contributions are three-fold: (i) invoking the notion of \textit{computation rate}, we develop a framework for energy-efficient \ac{ota-fl} with a \ac{pass}-assisted server. Our framework jointly tunes the \ac{pass} and schedules the devices, such that the energy required for model training at a desired aggregation accuracy is minimized. (ii) The energy-efficient \ac{ota-fl} design reduces to an NP-hard optimization. We develop a two-tier algorithm that approximates the solution of this problem at low-complexity. Our algorithm invokes alternating optimization to iterate among marginal problems whose solutions are determined efficiently. (iii) We validate our design through numerical experiments, where we compare the \ac{pass}-assisted \ac{ota-fl} design against the baseline \ac{ota-fl} with \ac{mimo} server. Our results demonstrate that, due to efficient compensation of large-scale effects, \ac{pass} can drastically reduce the energy required for \ac{ota-fl}, interestingly at a lower hardware cost. This places \acp{pass} as strong candidates for distributed learning in the next generations of wireless systems.

\paragraph*{Notation}
Throughout the paper, we show the scalars, vectors, and matrices with non-bold, bold lowercase, and bold uppercase letters, respectively. Expectation is denoted by $\Ex{\cdot}$, and the integer set $\set{1, \ldots,n}$ is shortened as $[n]$. 

\section{System Model}
Consider a distributed learning system, in which $K$ single-antenna devices communicate over wireless links with a single \ac{ps} that is equipped with a \ac{pass} transceiver. 
The location of device $k$ is $\mathbf{u}_k = [x_k, y_k, 0]$, where we assume that all devices are in the $xy$-plane. The \ac{pass} transceiver at the \ac{ps} consists of a waveguide of length $L$ with $N$ pinching elements, which can freely move along the waveguide. The waveguide is aligned along the $x$-axis at $y=0$ and altitude $z=a$. The coordinate of element $n$ on the waveguide is thus given by $\mathbf{v}_{n} = [\ell_{n}, 0, a]$, where $ 0 \leq \ell_{n} \leq L$ denotes the position of element $n$ on the waveguide. The signals are fed into and received from the waveguide at $x = 0$. 



\subsection{Distributed Learning Setting}
The edge device $k$ has a local dataset $\mathcal{D}_k = \{ (\mathbf{u}_{k,i}, v_{k,i}) : 1 \leq i \leq I_k \}$,
where $I_k = \abs{\mathcal{D}_k}$ denotes the dataset size, $\mathbf{u}_{k,i}$ is a data sample with $v_{k,i}$ being its corresponding label. The edge devices aim to collaboratively train a global model at the \ac{ps} without sharing their local data. To this end, they use the \ac{fedavg} algorithm: starting from a common model with parameter $\boldsymbol{\omega}$, device $k$ minimizes 
\begin{equation}
F_k(\boldsymbol{\omega}) = \frac{1}{I_k} \sum_{ (\mathbf{u}_k, v_k) \in \mathcal{D}_k } \ell \brc{\boldsymbol{\omega} \vert \mathbf{u}_{k}, v_{k}},
\end{equation}
using a gradient-based optimizer. Let $\boldsymbol{\omega}_k$ denote the locally-updated model. Device $k$ computes sufficient statistics 
$\bs_k$ from this local model and shares it with the \ac{ps}. Upon receiving all local updates, the \ac{ps} computes a weighted average as $\bs = \sum_{k=1}^K \varphi_k \bs_k$ for some predefined coefficients $\varphi_1, \ldots, \varphi_K$, and shares it back with the edge devices. The edge devices update their common model $\bomega$ using $\bs$ and repeat this procedure. We call each iteration a \textit{communication round}   
and use AirComp to aggregate over the air. 

\subsection{PASS-assisted AirComp Aggregation}
We now consider aggregation in a single round, in which the devices send their local parameters, i.e. $\mathbf{s}_{k}$ for $k\in \dbc{K}$, to the \ac{ps}. Using the \ac{ota-fl} framework, the devices transmit simultaneously their local parameters over shared time and frequency resources. Noting that the parameters are superimposed over the multiple access channel, we can realize the model aggregation directly over the air via proper analog pre- and post-processing. Without loss for generality, let us focus on transmission in a single time-frequency interval, in which the devices send a single model parameter $s_{k} \dbc{j}$. For sake of compactness, we drop the index $j$ in the sequel and represent a particular model parameter by $s_{k}$.

Prior to uplink transmission, the \ac{ps} schedules a subset $\setS \subseteq \dbc{K}$ of devices for model sharing. Device $k \in \setS$ scales its parameter $s_k$ as $x_k = b_k s_k$ for some $b_k \in \setC$ and transmits $x_k$ over the uplink channel. The average transmit power of each device is limited by $P$, i.e. $\abs{ b_k }^2\leq P$, assuming that $s_k$ is normalized prior to transmission.\footnote{This is typically the case in practice; see for instance \cite{yang2020federated}.} The signal received at the \ac{ps} is hence given by
\begin{align}\label{eqR}
y =\sum_{k} h_k\brc{\bl} \gamma_k x_k  + z,
\end{align}
where $\gamma_k \in \set{0,1}$ represents the scheduling status of device $k$, i.e. $\gamma_k = 1$ if $k\in \setS$ and $\gamma_k = 0$ otherwise, and $z$ is \ac{awgn} with mean zero and variance $\sigma^2$, i.e. $z \sim \mathcal{CN}(0, \sigma^2)$. The coefficient $h_{k} \brc{\bl}$ further denotes the channel between the \ac{pass} and device $k$, which depends on the location vector $\bl = \dbc{\ell_1, \ldots,\ell_N}$, representing the locations of elements on the waveguide. We will give the explicit model for $h_k\brc{\bl}$ later in this section. 

In an ideal scenario with noiseless links, the \ac{ps} aggregates the local parameters into a global parameter $\theta$ as
\begin{align}\label{eq:theta}
\theta=
\sum_{k} \varphi_k \gamma_k s_k.
\end{align}
Taking $\theta$ as the ground truth for aggregation, in noisy setting, the \ac{ps} estimates $\theta$ from its received signal $y$ directly in the analog domain as $\hat{\theta}=\rho y$ using a scaling factor $\rho\in\setC$. 


\subsection{PASS Channel Model}
The signal received from each device is the superposition of copies received via pinching elements and superimposed through the waveguide according to the relative distances of the elements. Following analysis in \cite{MIMO-PASS}, the uplink channel between the \ac{ps} and device $k$ is given by
\begin{align}\label{eq:g}
   {h}_{k} \brc{\bl} = \xi \alpha_{k}\sum_{n} 
    \frac{  
    \exp\set{-j \psi {E_{k}\brc{\ell_{n}}  } }
    }{ D_{k}\brc{\ell_{n}} },
\end{align}
where $\psi$ denotes the wavenumber computed from the wavelength $\lambda$ as $\psi = 2\pi/\lambda$, $\xi = {\lambda}/{4 \pi}$ is the coefficient proportional to the effective surface of the pinching elements, $\alpha_k$ represents the shadowing experienced by device $k$, $D_{k}\brc{\ell}$ is the distance between location $\ell$ on waveguide and device $k$, i.e. 
$D_{k}^2\brc{\ell} 
=(\ell -x_k)^2 + y_k^2 + a^2$, 
and $E_{k}\brc{\ell_{n}} = D_{k}\brc{\ell_{n}} + i_{\mathrm{ref}} \ell_{n}$ denotes the effective phase-shift at the signal received by element $\ell_n$ with $i_{\mathrm{ref}}$ being the reflective coefficient of the waveguide.

In the sequel, we can compactly write the received signal as
    $y = \bh \brc{\bl}^\trp \mGamma\mB\bs + z$,
where $\bh\brc{\bl} = \dbc{h_{1} \brc{\bl}, \cdots, h_{K} \brc{\bl}  }^\trp$ is the channel vector, $\mB=\diag{b_1, \cdots, b_K}$ contains the power scaling coefficients, $\mGamma=\diag{\gamma_1, \cdots, \gamma_K}$ is the scheduling matrix, and $\bs=[s_1, \cdots, s_K]^\trp$. 

\section{Energy Model for OTA-FL}
\ac{ota-fl} sketches a trade-off between aggregation quality and energy: by spending more energy 
the aggregation error at the \ac{ps} is reduced. Our ultimate goal is to develop an aggregation scheme that achieves a desired aggregation quality at minimum energy cost. There are various approaches to relate energy consumption to aggregation quality. In the sequel, we invoke a recent approach based on \textit{computation rate} \cite{ni22FL}. 

\subsection{Computation Rate}
Consider $\theta$ in \eqref{eq:theta} to be the target aggregation, which is estimated at the \ac{ps} by $\hat{\theta} = \rho y$. The \ac{mse} between the estimate and the target is given by
\begin{subequations}\label{eq:AggregationError}
\begin{align}
\varepsilon \brc{\bl,\mGamma,\mB,\rho}
&= \Ex{\abs{\hat{\theta}-\theta}^{2} }
\overset{\star}{=} \Ex{\abs{\rho y - \boldsymbol{\varphi}^\trp \mGamma\bs }^{2}}, \\
&\overset{\dagger}{=} \norm{\rho\bh(\bl)^\trp\mGamma\mB-\boldsymbol{\varphi}^{\trp}\mGamma }^{2} 
+\sigma^{2} \rho^2, 
\label{eq:agg2}
\end{align}
\end{subequations}
where we define $\boldsymbol{\varphi} = \dbc{\phi_1, \ldots,\phi_K}^\trp$ in the identity $\star$, and $\dagger$ uses $\Ex{\bs\bs^{\her}}=\mathbf{I}_{K}$ and independence of $z$. The computation rate is then defined in terms of the \ac{mse} as follows:

\begin{definition}[Computation rate \cite{ni22FL}]
    Let $\hat{\theta}$ be the over-the-air computation of $\theta$ with \ac{mse} $\varepsilon$ over a sub-channel with bandwidth $B$. The computation rate is then defined as
    \begin{align}
    \mathcal{R}_{\rm comp}= B \log_{2} \dbc{
    \mathbb{E} \{\vert{\hat{\theta}}\vert^{2}\}/\varepsilon
    }.
    \end{align}
\end{definition}

From the information-theoretic perspective, the computation rate characterizes the minimum rate required for transmission of $\hat{\theta}$ with distortion $\varepsilon$, assuming that $\hat{\theta}$ is distributed Gaussian with mean zero. An alternative interpretation is given by re-writing the computation-rate as
\begin{align}
\mathcal{R}_{\rm comp}= B \log_{2} \dbc{1+ 
\frac{1}{\varepsilon} \brc{
\mathbb{E} \{\vert{\hat{\theta}}\vert^{2}\}
- \varepsilon}
}.
\end{align}
In this form, $\mathcal{R}_{\rm comp}$ is the achievable rate over the channel used for analog computation, while treating the \ac{mse} as an uncorrelated Gaussian noise process. The key property of this metric is that it enables us to relate the transmission time to the aggregation error, as we will explain in the sequel.


\subsection{Energy Model}
In our setting, we have
\begin{align}
\mathbb{E} \{\vert{\hat{\theta}}\vert^{2}\} 
&= \Ex{\abs{\rho y}^{2}} 
=\abs{\rho}^2 \norm{\bh(\bl)\mGamma\mB}^{2}
+\abs{\rho}^2 \sigma^{2}.
\end{align}
We determine the computation \ac{snr} as
\begin{subequations}
\begin{align}
\mathrm{SNR}\brc{\bl,\mGamma,\mB,\rho}
&= \frac{
\mathbb{E} \{\vert{\hat{\theta}}\vert^{2}\}
}{
\varepsilon\brc{\bl,\mGamma,\mB,\rho}
}\\
&= \frac{ \abs{\rho}^2
\norm{\bh(\bl)^\trp\mGamma\mB}^{2} + \abs{\rho}^2 \sigma^{2}
}{
\norm{\rho\bh(\bl)^\trp\mGamma\mB - \boldsymbol{\varphi}^\trp \boldsymbol{\Gamma}}^{2}
+\abs{\rho}^2 \sigma^{2}
}.
\end{align}
\end{subequations}
Consequently, the computation rate is given by
\begin{align}
\mathcal{R}_{\mathrm{comp}}\brc{\bl,\mGamma,\mB,\rho}=B\log_{2}\brc{\mathrm{SNR}\brc{\bl,\mGamma,\mB,\rho}},
\end{align}
assuming that the deployed sub-channel has bandwidth $B$.


The notion of computation rate relates the aggregation error to the transmission duration: let $R_0$ denote the desired resolution for each local model in bits. This can be considered as the length of the quantized models in bits, whose quantization error is equivalent to the computation error achieved over the air. By setting the uplink transmission duration to
\begin{align} 
t \brc{\bl,\mGamma,\mB,\rho} = {R_0}/{
\mathcal{R}_{\mathrm{comp}}\brc{\bl,\mGamma,\mB,\rho}
},
\end{align}
it is guaranteed that the computation error is bounded by the quantization error at the desired resolution $R_0$. In this case, the transmission energy spent at device $k$ to share a model parameter at resolution $R_0$ is given by
\begin{align}
    \mathcal{E}_k \brc{\bl,\mGamma,\mB,\rho} = \gamma_k \abs{b_k}^2 t\brc{\bl,\mGamma,\mB,\rho}.
\end{align}

\section{Energy Efficient OTA-FL Scheme}
We design an energy-efficient scheme, which addresses the following requirements: it ($i$) minimizes the total transmission energy, and ($ii$) guarantees a threshold learning performance. 
This is formulated through the following optimization.
\begin{align}
    \min_{\bl,\mGamma,\mB,\rho} \quad 
    & \sum_{k=1}^K \mathcal{E}_k\brc{\bl,\mGamma,\mB,\rho}
    \tag{$\mathcal{P}$}
    \label{eq:P}
    \\
    \text{s.t.} \quad 
    & |b_k|^2 \le P,\quad \forall \; k, 
    \tag{$\mathrm{C}_1$}
    \label{eq:C1}
    \\
    & \gamma_k \in \{0,1\}, \quad \;\,\forall \; k, 
    \tag{$\mathrm{C}_2$}
    \label{eq:C2}
    \\
    & \sum_{k=1}^K \gamma_k \ge K_{\min}, 
    \tag{$\mathrm{C}_3$}
    \label{eq:C3}
    \\
    & 0 \le \ell_{n} \le L \; \text{ and } \; 
    \ell_{n+1} -\ell_{n} \ge \Delta , \quad \forall \; n.
    \tag{$\mathrm{C}_4$}
    \label{eq:C4}
\end{align}
In \eqref{eq:P}, the objective determines the total transmission energy, whose minimization addresses requirement ($i$). The constraint \eqref{eq:C3} further addresses requirement ($ii$) by guaranteeing that at least $K_{\min}$ devices participate in the \ac{fedavg} algorithm. 
The remaining constraints impose system restrictions, namely \eqref{eq:C1} limiting transmit power, and \eqref{eq:C4} restricting the \ac{pass} element locations to their physically feasible region. 

Noting that the transmission time is the same for all devices, we can simplify \eqref{eq:P} to
\begin{align}
    \min_{\bl,\mGamma,\mB,\rho}  \quad 
    & \frac{
    \sum_k \abs{b_k}^2 \gamma_k
    }{
    \log_2\brc{
    \mathrm{SNR}\brc{\bl,\mGamma,\mB,\rho}
    }
    }
    \tag{$\mathcal{P}'$}
    \label{eq:Pp}
    \quad
    \text{s.t.} \quad (\mathrm{C}_1)-(\mathrm{C}_4) .
\end{align}
This problem is a mixed program: the scheduling variables create a combinatorial search, and the other variables show up nonlinearly in the objective rendering non-convexity. 

\paragraph*{Proposed Solution in Nutshell}
We tackle \eqref{eq:Pp} using alternating optimization: the problem is broken into three \textit{marginal problems,} where each problem optimizes the objective with respect to a subset of variables, treating the others as fixed. For each marginal problem, we develop a low-complexity algorithm that solves the problem efficiently. The joint solution is then estimated within an outer tier that alternates between the marginal solutions. %
In the sequel, we discuss each marginal problem separately. 

\subsection{PASS Transceiver Design}
The first marginal problem is specified by fixing $\mGamma$ and $\mB$ in \eqref{eq:Pp}, and optimizing the objective with respect to the positions of the elements, i.e. $\bl$, and the receiver scaling factor $\rho$. The objective in this case depends on the optimization variables through the denominator. We can hence re-write the problem as a maximization of the denominator. Noting that only \eqref{eq:C4} depends on $\bl$ and the fact that $\log$ is an increasing function, the marginal problem reduces to
\begin{align}
\max_{\bl,\rho}\quad 
\mathrm{SNR}\brc{\bl,\mGamma,\mB,\rho}
    \tag{$\mathcal{M}_1$} \label{eq:M1}
\quad
\text{s.t.}\quad 
(\mathrm{C}_4).
\end{align}
This problem is non-convex due to the constraint \eqref{eq:C4}. A basic approach is to solve it by a step-wise scheme. We however note that the constraint \eqref{eq:C4} only restricts the subspace of $\setC^K$ that is represented by $\bh\brc{\bl}$. We can use this fact to solve the problem more efficiently via a projection-based approach, as illustrated in the sequel.

For a given set of devices $\setS$, let 
$\tilde{\bh}(\bl) = \dbc{h_k\brc{\bl}: k \in \setS}$,
$\tilde{\mB} = \diag{b_k: k \in \setS}$, and
$\tilde{\boldsymbol{\varphi}} = \dbc{\varphi_k: k \in \setS}$.
We can then write the objective function of \eqref{eq:M1} as
\begin{align}\label{eq:F_v_rho}
F\brc{\bv,\rho} = \frac{
        \norm{\rho \bv}^{2} + \sigma^{2} \abs{\rho}^2
        }{
        \norm{\rho \bv - \btph}^{2}
        +\sigma^{2} \abs{\rho}^2
        },
\end{align}
where $\bv = \tilde\mB \tilde\bh\brc{\bl}$. Noting that the length of the waveguide is significantly larger than the wavelength, one can argue that the phase of the entries of $\bv$ can be arbitrarily tuned. Nevertheless, the norm of $\bv$ is restricted by the transmit power and the path-loss. In fact, we can write
\begin{align}
    \norm{\bv}^2 = \norm{\tilde\mB\bh\brc{\bl}}^2 &= \sum_{k \in \setS} \abs{h_k\brc{\bl}}^2\abs{b_k}^2\leq A \Vert\tilde\mB\Vert^2,
\end{align}
where $A$ is the upper-bound on the norm of $\tilde{\bh}(\bl)$ given by
\begin{align} \label{eq:BoundA}
    A = \sum_{k \in \setS} \frac{\abs{\xi\alpha_k}^2 N}{D_k^2\brc{x_k}},
\end{align}
with $ x_k$ being the $x$-coordinate of $\bu_k$. 

The above observation suggests that the solution to \eqref{eq:M1} can be accurately approximated by $\rho^\star$ and $\bv^\star$, which are the solutions to the following optimization  
\begin{align}\label{eq:M1-p}
\max_{\bv,\rho}\quad 
F\brc{\bv,\rho}
\quad
\text{s.t.}\quad 
\norm{\bv} \leq Q,
\end{align}
with $F$ given in \eqref{eq:F_v_rho} and $Q = \sqrt{A}\Vert \tilde\mB \Vert$. Knowing the solution $(\rho^\star, \bv^\star)$ to \eqref{eq:M1-p}, one can set $\rho = \rho^\star$ and $\bl $ to what minimizes $\Vert{\tilde{\mB} \bh\brc{\bl^\star} - \bv^\star}\Vert$. More precisely, we can approximate the solution of \eqref{eq:M1} by finding the solution of the relaxed problem in \eqref{eq:M1-p} and then projecting it into the space represented by $\bh\brc{\bl}$. The following theorem gives the solution to \eqref{eq:M1-p}.

\begin{theorem}
    The solution to \eqref{eq:M1-p} is given by $\rho^\star = {\norm{\btph}}/{Q}$ and $\bv^\star = \brc{{Q}/{\norm{\btph}}}\btph$.
\end{theorem}

\begin{IEEEproof}
    We first note that the phase of $\rho$ only appears in the denominator adding a shift to the phase of $\bv$. As the phase of $\bv$ is not restricted, we conclude that $\rho$ can be limited to the real axis. Next, we decompose $\bv$ in terms of its amplitude $0 \leq \zeta \leq Q$ and directional vector $\be \in \setC^K$, which satisfies $\norm{\be} = 1$. In this case, the optimization is rewritten as
\begin{align}
\max_{\be,\zeta,\rho}\quad 
&\frac{
        \rho^2 \zeta^2 + \sigma^{2} \rho^2
        }{
        \norm{\rho \zeta \be - \btph}^{2}
        +\sigma^{2} \rho^2
        }
\quad
\text{s.t.}\quad 
0\leq \zeta \leq Q.
\end{align}
The objective is maximized in $\be$ by setting $\be$ to be co-linear with $\btph$, i.e. $\be^\star = \btph / \norm{\btph}$. Defining $a = \rho\zeta$, we can solve 
\begin{align}\label{eq:FQ}
\max_{a,\rho}\quad 
&\frac{
        a^2 + \sigma^{2} \rho^2
        }{
        \brc{a - \norm{\btph}}^{2}
        +\sigma^{2} \rho^2
        }
\quad
\text{s.t.}\quad 
0\leq a \leq Q \rho,
\end{align}
which is quadratic in $a$. The solution is thus given by either the unconstrained maximizer or the upper bound, i.e. $a = Q\rho$. We can hence write the solution as
\begin{align}\label{eq:AA}
    a^\star\brc{\rho} = 
    \begin{cases}
    Q\rho &\rho \geq \rho_{\mathrm{th}}\\
    \tfrac{\norm{\btph}}{2} \brc{1 + \sqrt{1+4\tfrac{\sigma^2}{\norm{\btph}^2} \rho^2}} &\rho < \rho_{\mathrm{th}}
    \end{cases},
\end{align}
where $\rho_{\mathrm{th}} = + \infty$ when $Q\leq \sigma$ and $\rho_{\mathrm{th}} = {Q \norm{\btph}}/\brc{Q^2 - \sigma^2}$ when $Q > \sigma$. 
Replacing $a^\star\brc{\rho}$ in the objective of \eqref{eq:FQ}, we can rewrite the objective in terms of $\rho$ as
\begin{align}
F\brc{\rho}=\frac{
        a^\star\brc{\rho}^2 + \sigma^{2} \rho^2
        }{
        \brc{a^\star\brc{\rho} - \norm{\btph}}^{2}
        +\sigma^{2} \rho^2
        }.
\end{align}
The optimizer is then given by
\begin{align}
    \max_{\rho \geq 0} F\brc{\rho} = \max \set{ 
    \max_{0 \leq \rho < \rho_{\mathrm{th}}} F\brc{\rho}, 
    \max_{\rho \geq \rho_{\mathrm{th}}} F\brc{\rho}
    }.
\end{align}
For each of the inner optimizations, we can use \eqref{eq:AA} to write
\begin{subequations}
\begin{align}
    &\argmax_{0 \leq \rho < \rho_{\mathrm{th}}} F\brc{\rho} =  {\norm{\btph}}/{Q}\\
    &\argmax_{\rho \geq \rho_{\mathrm{th}}} F\brc{\rho} = \rho_{\mathrm{th}}.
\end{align}
\end{subequations}
We next note that $\rho_{\mathrm{th}} > {\norm{\btph}}/{Q}$ and that $F\brc{\rho}$ is concave for ${0 \leq \rho < \rho_{\mathrm{th}}}$ and decaying for ${\rho \geq \rho_{\mathrm{th}}}$. This indicates that the maximizer in ${0 \leq \rho < \rho_{\mathrm{th}}}$ is larger than the one in ${\rho \geq \rho_{\mathrm{th}}}$, concluding that $\rho^\star = {\norm{\btph}}/{Q}$ and $a^\star = {\norm{\btph}}$. 
\end{IEEEproof}

From \eqref{eq:M1-p}, the solution to \eqref{eq:M1} is approximated via the Gauss-Seidel scheme: starting with an initial $\bl$, we update each $\ell_n$ separately to minimize $\Vert{\tilde{\mB} \bh\brc{\bl^\star} - \bv^\star}\Vert$ with others treated as fixed. The algorithm is summarized in Algorithm~\ref{alg:1}. Here, 
    $f_n\brc{\ell_n}=\Vert{\mGamma{\mB} \bh\brc{\bl\brc{\ell_n}} - \bv}\Vert$, 
where $\bl\brc{\ell_n}$ denotes $\bl$ when all entries except $\ell_n$ being fixed, and \texttt{Grid}$\brc{f\vert A:B}$ denotes grid search in $\dbc{A,B}$ for minimizer of $f$. 
\begin{algorithm}[t!]
\caption{PASS Tuning}\label{alg:1}
\begin{algorithmic}[1]
\State Set $Q = \sqrt{A} \Vert\mB \mGamma\Vert$, 
$\rho = \norm{\mGamma\boldsymbol{\varphi}} / Q$, and $\bv = {\mGamma\boldsymbol{\varphi}} /\rho$
\State Set $\ell_0 = 0$, $\ell_{N+1} = L$, $t=0$, and $\varepsilon_0$ to a small constant
\While{$\varepsilon > \varepsilon_0$ and $t < T$}
\For{$n=1:N$}
\State Update $\ell_n \gets \texttt{Grid}\brc{f_n\vert \ell_{n-1} + \Delta: \ell_{n+1} - \Delta}$
\EndFor
\State Update $t \gets t+1$ and $\varepsilon \gets \Vert{\mGamma{\mB} \bh\brc{\bl} - \bv}\Vert$
\EndWhile
\end{algorithmic}
\end{algorithm}

\subsection{Transmit Power Scaling}
The second marginal problem solves \eqref{eq:P} in terms of $\mB$ for a fixed \ac{pass} configuration and scheduled subset $\setS$. The marginal problem in this case is given by 
\begin{align}
    \min_{\mB}  \quad 
    & {
    \norm{\mGamma\mB}^2
    }/{
    \log_2\brc{\mathrm{SNR}\brc{\bl,\mGamma,\mB,\rho}
    }
    } 
    \tag{$\mathcal{M}_2$}
    \label{eq:M2}
  \quad \text{s.t.} \quad (\mathrm{C}_1). 
\end{align}

The problem in \eqref{eq:M2} describes generalized fractional programming, in which the nominator is convex and the denominator is difference of concave functions. While this problem is hard to solve directly, we can adopt the \emph{Dinkelbach’s method} to approximate its solution iteratively: starting with an initial $\mB^0$, we update $\mB^{t+1}$ from $\mB^t$ by solving 
\begin{align}
    \min_{\mB}  
    & \norm{\mGamma\mB}^2 - \eta^t{
    \log_2\brc{
    \mathrm{SNR}\brc{\bl,\mGamma,\mB,\rho}
    }
    } 
    \label{eq:M2-DC}
  \quad\text{s.t.} \quad (\mathrm{C}_1), 
\end{align}
where 
$
\eta^t = {\norm{\mGamma\mB^t}^2}/{
    \log_2\brc{\mathrm{SNR}\brc{\bl,\mGamma,\mB^t,\rho}
    }
    }
$. 
The problem in \eqref{eq:M2-DC}, describes a difference-of-convex optimization, in which the difference of two convex terms is minimized on a convex set. We hence follow classical minorization minimization \cite{ni22FL} to approximate its solution with the minimizer of the surrogate function that is determined by replacing the concave term, i.e. the second term in the difference, with its first-order truncated Taylor series. Due to the lack of space, we skip the details and refer the reader to \cite{ni22FL}. Using this surrogate minimization, the solution of \eqref{eq:M2-DC} is approximated as
\begin{align}\label{eq:BB2}
    b_k^{t+1} = \mathcal{P}_P\brc{\frac{
    \rho h_k(\bl)  \varphi_k/{
     \kappa^t
    }
    }{
    {\ln 2}/{\eta^t}- \abs{\rho h_k(\bl)}^2  \brc{{1}/{\tau^t}- {1}/{\kappa^t}}
    }},
\end{align}
for $k\in \setS$, where $\mathcal{P}_P \brc{x} = x$ for $\abs{x}^2 \leq P$ and $\mathcal{P}_P \brc{x} = {\sqrt{P}x}/{\abs{x}}$ if $\abs{x}^2 > P$, 
and $\kappa^t$ and $\tau^t$ are given by
\begin{subequations}\label{eq:Updt2}
\begin{align}
    \tau^t &=\norm{\rho\bh\brc{\bl}^\trp\mGamma \mB^t}^{2}
        +\sigma^{2} \rho^2, \\
        \kappa^t &= \norm{\rho \bh \brc{\bl}^\trp\mGamma\mB^t - \tilde{\boldsymbol{\varphi}}^\trp\mGamma}^{2}
        +\sigma^{2} \rho^2 .
\end{align}
\end{subequations}
The solution of \eqref{eq:M2} is then approximated by iteratively updating $b^{t}_k$ till convergence, as summarized in Algorithm~\ref{alg:2}.

\begin{algorithm}[t!]
\caption{Transmit Power Scaling}\label{alg:2}
\begin{algorithmic}[1]
\State Set $t=0$ and set $\varepsilon_0$ to a small constant
\While{$\varepsilon > \varepsilon_0$ and $t < T$}
\State Update $\eta^{t}$ and $\brc{\kappa^{t},\tau^{t}}$ 
\State \textbf{for} $k\in\setS$ \textbf{do} Update $b_k^{t+1}$ with \eqref{eq:BB2} \textbf{end for}
\State Update $t \gets t+1$ and $\varepsilon \gets \norm{\mB^{t+1} \mGamma - \mB^t \mGamma}$
\EndWhile
\end{algorithmic}
\end{algorithm}

\subsection{Device Scheduling}
The last marginal problem addresses the device selection task for a given power scaling and \ac{pass} tuning. This problem finds a similar form as \eqref{eq:M2}, with the key difference that the optimization variables are binary in this case. While one can extend Dinkelbach's method to this problem, we take a more practical approach based on step-wise regression. To this end, let us rewrite the marginal problem directly in terms of the selection subset $\setS$ as
\begin{align}
\min_{\setS}\quad 
& G \brc{\setS}    
       \tag{$\mathcal{M}_3$} \label{eq:M3}
\quad
\text{s.t.}\quad 
\abs{\setS} \geq K_{\min} ,
\end{align}
where the marginal objective $G$ is defined as
\begin{align}
G \brc{\setS} = 
\frac{\sum_{k\in \setS}  \abs{b_k}^2}{\log_2\brc{1+ \sum_{k\in \setS} C_k^0} - \log_2\brc{1+ \sum_{k\in \setS} C_k^1}
},
\end{align}
with $C_k^m$ being defined for $m\in\set{0,1}$ as
\begin{align}
C_k^m = 
    \norm{\bh(\bl)^\trp\mGamma\mB / \sigma - m \boldsymbol{\varphi}^\trp \boldsymbol{\Gamma} / \rho \sigma}^{2}.
\end{align}
A fast approximate solution to \eqref{eq:M3} is given by iterative select and reject: starting with $\setS^0 = [K]$, we first drop devices step-wise. In iteration $t$, we select the device $i^t \in \setS^t$ such that the objective drop $G(\setS^t ) - G(\setS^t - \set{i^t}) $ is maximized. The dropping stops, when either no further objective reduction is possible or $\abs{\setS^t} = K_{\min}$. In the former case, we improve our selection by exchanging $i \in \setS^t$ with $j\notin \setS^t$ to check for potential objective reduction, and repeat this until convergence. This step-wise approach is summarized in Algorithm~\ref{alg:3}. In this algorithm, we define $\mathrm{d} G_{i,j} = G(\setS ) - G([\setS - \set{i}] \cup \set{j})$.
\begin{algorithm}[t]
\caption{Device Scheduling}\label{alg:3}
\begin{algorithmic}[1]
\State Initiate $\setS = [K]$ and set \texttt{flag = True}
\While{$\abs{\setS} > K_{\min}$ and \texttt{flag = True} }
\State Set $i^\star \gets\argmax_{i \in \setS} G(\setS ) - G(\setS - \set{i}) $ 
\If{$G(\setS ) > G(\setS - \set{i^\star}) $}
\State Update $\setS \gets \setS - \set{i^\star}$
\Else
\State \textbf{for} {$i\in \setS$} \textbf{do} $j^\star \brc{i} = \argmax_{j \notin \setS} 
\mathrm{d} G_{i,j}
$ 
\textbf{end for}
\State Set $i^\star \gets\argmax_{i \in \setS} 
\mathrm{d} G_{i,j^\star(i)}
$
\If{$\mathrm{d} G_{i^\star,j^\star(i^\star)} > 0
$}
\State Update $\setS \gets [\setS - \set{i^\star}] \cup \set{j^\star(i^\star)}$
\State \hspace{-6.5mm} \textbf{else} Set \texttt{flag = False} 
\EndIf
\EndIf
\EndWhile
\end{algorithmic}
\end{algorithm}

\subsection{Two-Tier Algorithm and Its Complexity}
The solution to the joint optimization \eqref{eq:P} is approximated by an outer loop, which alternates among Algorithms \ref{alg:1}-\ref{alg:3} for a maximum number of iterations or until the solution converges. The complexity of the algorithm is thus scaled as $\mathcal{O}\brc{TC_{N,K}}$ with $T$ being the number of iterations in the outer loop and $C_{N,K}$ being the complexity order of the inner-loop algorithms. Assuming Algorithm~\ref{alg:1} to iterate for $T_1$ iterations, \ac{pass} tuning imposes a complexity of order $\mathcal{O}\brc{T_1 N}$. The complexity of Algorithm~\ref{alg:2} is also of order $\mathcal{O}\brc{T_2 K}$ for $T_2$ iterations, and the step-wise scheduling imposes the complexity of $\mathcal{O}\brc{K \log K}$. Our numerical experiments confirm that the inner loops and the outer loop converge within a constant number of iterations, imposing low computational complexity on the \ac{ps}.

\section{Numerical Experiments}
We validate our design through numerical experiments. We train a four-layer convolutional neural network on MNIST in a network with $K=8$ devices. The training data is distributed uniformly among the devices, and in each communication round at least $K_{\min} = 6$ devices are active. We use \ac{fedavg} with 2 local epochs at the devices. The devices are uniformly distributed within a square region with edge length $D = 50$ m. The \ac{pass} transceiver is located at height $a=5$ m and is pinched with $N=32$ elements. Unless stated otherwise, the following is considered: the noise power at the \ac{ps} is $ \log \sigma^2 = -90$ dBm, carrier frequency is $f_0 = 5$ GHz, \textit{transmit} power at devices is $\log P$ = 0 dBm, shadowing coefficients are set to $\alpha_k= 1$, and the distance between the pinching elements is set to be larger than $\Delta = \lambda/2$. As the baseline, we consider (i) the case with perfect links, i.e. vanilla \ac{fedavg}, and (ii) a setting in which the \ac{ps} is equipped with a \ac{mimo} transceiver with $M$ antennas at the center of the square. We consider line-of-sight channels between the devices and \ac{ps}. It is worth noting that hardware-wise, this setting is comparable with the \ac{pass} only if it has $M=1$ antennas, as the \ac{ps} requires $M$ separate radio-frequency chains.

\begin{figure}[t!]
  \centering
  \begin{tikzpicture}
    \begin{axis}[
        width=3in, height=1.8in,
        xlabel={round},
        ylabel={Accuracy},
        ymin=5, ymax=106,
        xmin=1, xmax=42,
        xtick={10, 20, 30, 40},
        ytick={20, 60, 100},
        legend style={
            at={(.9,0.5)},
            anchor=east
        },
    ]

\addplot[blue, thick, mark = triangle] table[row sep=\\, x index=0, y index=1] {
1  83.24 \\
2  89.24 \\
3  93.74 \\
4  90.60  \\
5  93.71 \\
6  94.29\\
7  93.91  \\
8  95.57  \\
9  94.45  \\
10  94.57 \\
11  96.48 \\
12  94.37  \\
13  96.14 \\
14  95.73 \\
15  96.08 \\
16  94.97\\
17  96.46 \\
18  94.17  \\
19  96.24 \\
20  95.84  \\
21  94.57  \\
22  95.86 \\
23  94.30  \\
24  94.39 \\
25  96.76 \\
26  96.47 \\
27  95.54 \\
28  96.39 \\
29  95.90 \\
30  96.56 \\
31  94.20 \\
32  94.81 \\
33  96.89 \\
34  96.94 \\
35  96.06 \\
36  96.55 \\
37  96.79 \\
38  96.62 \\
39  96.76 \\
40  96.79 \\
};
\addlegendentry{\footnotesize PASS}

\addplot[red, thick] table[row sep=\\, x index=0, y index=1] {
1 88.43 \\
2 89.68 \\
3 94.15 \\
4 94.58 \\
5 95.53 \\
6 96.52 \\
7 96.68 \\
8 96.43 \\
9 97.32 \\
10 97.23 \\
11 97.45 \\
12 97.54 \\
13 97.65 \\
14 97.52 \\
15 97.82 \\
16 97.97 \\
17 98.17 \\
18 98.12 \\
19 98.05 \\
20 98.14 \\
21 98.50 \\
22 98.43 \\
23 98.16 \\
24 98.50 \\
25 98.44 \\
26 98.71 \\
27 98.36 \\
28 98.63 \\
29 98.67 \\
30 98.52 \\
31 98.66 \\
32 98.70 \\
33 98.86 \\
34 98.77 \\
35 98.61 \\
36 98.85 \\
37 98.88 \\
38 98.88 \\
39 98.91 \\
40 98.45 \\
};
\addlegendentry{\footnotesize FedAvg}

\addplot[green!40!black, thick, mark=*] table[row sep=\\, x index=0, y index=1] {
1 11.26 \\
2 10.32 \\
3 10.10 \\
4 10.28 \\
5 10.10 \\
6 9.74 \\
7 10.28 \\
8 9.82 \\
9 8.92 \\
10 9.48 \\
11 9.74 \\
12 10.28 \\
13 10.09 \\
14 10.09 \\
15 10.09 \\
16 9.58 \\
17 10.28 \\
18 11.35 \\
19 9.58 \\
20 11.35 \\
21 11.35 \\
22 10.09 \\
23 10.28 \\
24 10.09 \\
25 10.10 \\
26 10.09 \\
27 11.35 \\
28 10.10 \\
29 10.09 \\
30 10.32 \\
31 11.35 \\
32 10.32 \\
33 9.58 \\
34 10.28 \\
35 10.32 \\
36 9.82 \\
37 10.27 \\
38 10.28 \\
39 9.74 \\
40 9.82 \\
};
\addlegendentry{\footnotesize MIMO - $M=8$}

\addplot[green!40!black, thick, mark=o] table[row sep=\\, x index=0, y index=1] {
1 8.92 \\
2 10.28 \\
3 8.92 \\
4 9.80 \\
5 8.92 \\
6 9.74 \\
7 11.35 \\
8 10.28 \\
9 11.35 \\
10 9.58 \\
11 10.28 \\
12 10.28 \\
13 11.35 \\
14 10.10 \\
15 10.09 \\
16 10.10 \\
17 10.10 \\
18 9.82 \\
19 9.58 \\
20 11.35 \\
21 9.82 \\
22 9.82 \\
23 10.10 \\
24 9.74 \\
25 9.82 \\
26 9.80 \\
27 9.82 \\
28 10.09 \\
29 10.28 \\
30 10.28 \\
31 11.35 \\
32 9.82 \\
33 10.08 \\
34 10.32 \\
35 10.32 \\
36 9.82 \\
37 9.80 \\
38 10.09 \\
39 11.35 \\
40 11.35 \\
};
\addlegendentry{\footnotesize MIMO - $M=32$} 
\end{axis}
\end{tikzpicture}
  \caption{Accuracy vs communication round at $D=50$ m and $\log P = 0$ dBm. The MIMO \ac{ps} is unable to learn due to low receive SNR.}
  \label{fig:1}
\end{figure}
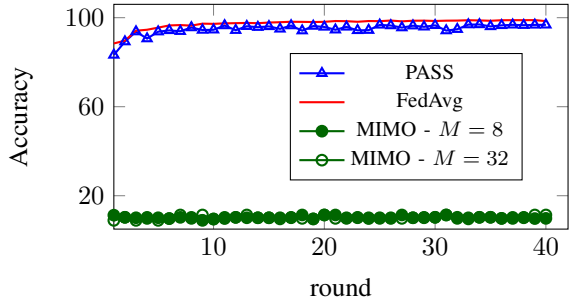
Fig.~\ref{fig:1} shows the results, comparing the \ac{pass} against the \ac{mimo} system for different choices of $M$. As the figure shows, while \ac{pass} closely tracks the learning curve for \ac{fedavg}, the \ac{mimo} \ac{ps} is not trained within the first 40 epochs even with $M=32$ transmit antennas. This behavior is intuitive, as the extension of the waveguide across the region enables the \ac{pass} to substantially compensate for the path-loss, and therefore improves the receive \ac{snr} drastically.

\begin{figure}[t!]
  \centering
  \begin{tikzpicture}
    \begin{axis}[
        width=3in, height=1.8in,
        xlabel={$\log D$},
        ylabel={Accuracy},
        ymin=5, ymax=89,
        xmin=9, xmax=500,
        xmode=log,
        legend pos=south west,
    ]

\addplot[blue, thick, mark = triangle] table[row sep=\\, x index=0, y index=1] {
10  84.54\\
50  84.33\\
100  83.91\\
200  83.71\\
400  83.38\\
};
\addlegendentry{\footnotesize PASS}

\addplot[green!40!black, thick, mark=square] table[row sep=\\, x index=0, y index=1] {
10    80.63\\
50   51.46\\
100    43.53\\
200    32.50\\
400    22.46\\
};
\addlegendentry{\footnotesize $M=8$}

\addplot[green!40!black, thick, mark=*] table[row sep=\\, x index=0, y index=1] {
10      83.40 \\
50     70.43 \\
100   52.52 \\
200     40.10 \\
400     28.11 \\
};
\addlegendentry{\footnotesize $M=16$}

\addplot[green!40!black, thick, mark=o] table[row sep=\\, x index=0, y index=1] {
10   83.45\\
50  75.67\\
100  58.31\\
200    42.52\\
400   23.07\\
};
\addlegendentry{\footnotesize $M=32$}
    
    \end{axis}
    \end{tikzpicture}
  \caption{Accuracy vs area size. The results are reported after training a two-layer MLP for 20 rounds at $\log P = 0$ dBm.}
  \label{fig:2}
\end{figure}
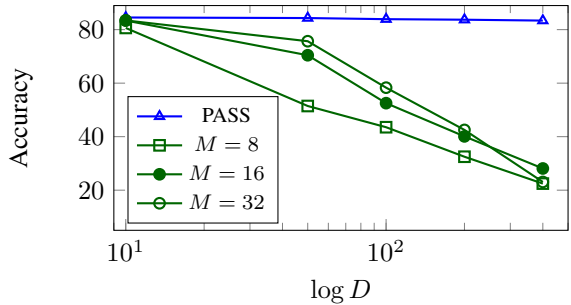
To elaborate the order of path-loss compensation achieved by the \ac{pass}, we next plot the final accuracy after 20 rounds against the region edge size $D$ in Fig.~\ref{fig:2}. For better demonstration, we consider MNIST classification with a basic two-layer multilayer perceptron (MLP), as the training performance in this case scales smoothly with the aggregation error due to the limited model capacity. As the result depicts, the \ac{pass} offers an almost constant training performance across $D$, while the accuracy of the \ac{mimo} system drops considerably as $D$ increases. Fig.~\ref{fig:3} shows the accuracy against the transmit power for the same setting as in Fig.~\ref{fig:2} at $D=50$ m. The results show a significant gain in terms of transmit power, demonstrating the significant path-loss compensation offered by the \ac{pass}.

\begin{figure}[t!]
  \centering
  \begin{tikzpicture}
    \begin{axis}[
        width=3in, height=1.8in,
        xlabel={$\log P$ [dBm]},
        ylabel={Accuracy},
        ymin=5, ymax=89,
        xmin=-25, xmax=65,
        xtick={-20, -10, 0, 10, 20, 30, 40, 50, 60},
        legend pos=south east,
    ]

\addplot[blue, thick, mark = triangle] table[row sep=\\, x index=0, y index=1] {
-20 79.51 \\
-10 84.51 \\
0 84.63 \\
10 84.83 \\
20 84.89 \\
30 84.72 \\
40 84.56 \\
50 84.17 \\
60 84.71 \\
};
\addlegendentry{\footnotesize PASS}

\addplot[green!40!black, thick, mark=square] table[row sep=\\, x index=0, y index=1] {
-20 12.06 \\
-10 29.48 \\
0 57.54 \\
10 72.70 \\
20 82.11 \\
30 83.04 \\
40 82.83 \\
50 83.03 \\
60 82.95 \\
};
\addlegendentry{\footnotesize $M=8$}

\addplot[green!40!black, thick, mark=*] table[row sep=\\, x index=0, y index=1] {
-20 21.20 \\
-10 39.44 \\
0 70.08 \\
10 82.77 \\
20 82.94 \\
30 83.15 \\
40 83.14 \\
50 83.21 \\
60 83.14 \\
};
\addlegendentry{\footnotesize $M=16$}

\addplot[green!40!black, thick, mark=o] table[row sep=\\, x index=0, y index=1] {
-20 25.11 \\
-10 44.97 \\
0 67.14 \\
10 81.88 \\
20 82.96 \\
30 83.09 \\
40 82.80 \\
50 83.00 \\
60 83.25 \\
};
\addlegendentry{\footnotesize $M=32$}
    
    \end{axis}
    \end{tikzpicture}
   \caption{Accuracy vs area size. The results are reported after training a two-layer MLP for 20 rounds at $D = 50$ m.}
  \label{fig:3}
\end{figure}
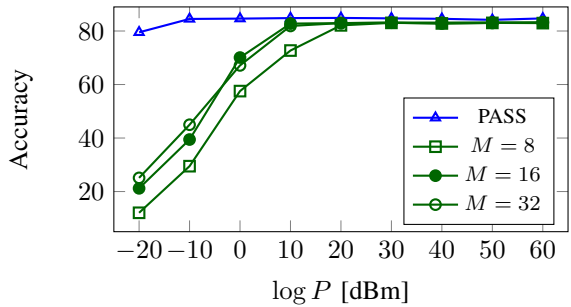

\section{Conclusions}
This work studied \ac{ota-fl} via a \ac{pass}-assisted \ac{ps}. We developed an algorithmic approach, which jointly schedules the devices and tunes the \ac{pass} at the \ac{ps}. The results demonstrate that using the \ac{pass}, the required transmit power for efficient aggregation is significantly reduced. This gain is mainly due to effective path-loss compensation realized by the flexibility of the \ac{pass}. The findings of this study suggest that \ac{pass}-assisted transmission can boost significantly the energy efficiency of distributed learning in wireless settings with \ac{los} links between devices and the \ac{ps}, e.g. indoor systems. Extending the results of this study to \ac{mimo}-\ac{pass} transceiver is a natural direction for future work. The work in this direction is currently ongoing.

\bibliographystyle{IEEEtran}
\bibliography{ref.bib}
\end{document}